\documentclass[11pt]{article}
\usepackage{geometry}
\usepackage{fullpage}

\usepackage{amsmath}
\usepackage{amssymb}
\usepackage{amsthm}

\usepackage[ruled]{algorithm2e}

\usepackage[backref=page]{hyperref}
\hypersetup{
	bookmarks=true,
	bookmarksnumbered=true,
	bookmarksopen=true,
	pdftitle={One Tree Suffices: A Simultaneous O(1)-Approximation for Single-Sink Buy-at-Bulk},
	pdfauthor={Ashish Goel, Ian Post}
}

\newtheorem{thm}{Theorem}[section]
\newtheorem{lem}[thm]{Lemma}
\newtheorem{cor}[thm]{Corollary}

\newtheorem{defn}[thm]{Definition}

\newcommand{\E}{\textrm{\textbf{E}}}

\begin{document}

\title{One Tree Suffices: A Simultaneous $O(1)$-Approximation for Single-Sink Buy-at-Bulk}
\author{Ashish Goel
\thanks{Departments of Management Science and Engineering, and by courtesy,
Computer Science, Stanford University. Email: \href{mailto:ashishg@stanford.edu}{ashishg@stanford.edu}.
Research funded by an NSF IIS grant, by funds from Google, Microsoft, and Cisco, a 3-COM faculty fellowship, and by the ARL Network Science CTA.}
\\Stanford University
\and Ian Post
\thanks{Department of Computer Science, Stanford University. Email: \href{mailto:itp@stanford.edu}{itp@stanford.edu}.  
Research funded by an NSF IIS grant.}
\\Stanford University}


\maketitle

\pdfbookmark[1]{Abstract}{MyAbstract}
\begin{abstract}
We study the single-sink buy-at-bulk problem with an unknown cost function.
We wish to route flow from a set of demand nodes to a root node, where the cost of routing $x$ total flow along an edge is proportional to $f(x)$ for some concave, non-decreasing function $f$ satisfying $f(0)=0$.  
We present a simple, fast, combinatorial algorithm that takes a set of demands and constructs a single tree $T$ such that for all $f$ the cost $f(T)$ is a 47.45-approximation of the optimal cost for that $f$.  
This is within a factor of 2.33 of the best approximation ratio currently achievable when the tree can be optimized for a specific function.
Trees achieving simultaneous $O(1)$-approximations for all concave functions were previously not known to exist regardless of computation time.
\end{abstract}

\section{Introduction}

Many natural network design settings exhibit some form of economies of scale that reduce the costs when many flows are aggregated together.  
We may benefit from cheaper bandwidth when laying high capacity network links \cite{andrews1998access}, reduced infrastructure costs or bulk discounts for shipping large amounts of goods together \cite{salman1997bbn}, or summarization and compression of correlated information flows  \cite{krishnamachari2002modelling}.  
These scenarios are known in the literature as buy-at-bulk problems.  In a general buy-at-bulk problem we are given a graph and a set of demands for flow between nodes.  
The cost per unit length for routing a total of $x$ flow along an edge is $f(x)$ for some function $f$.  To model the economies of scale, we assume $f$ is concave and monotone non-decreasing.


We will focus on single-sink (or single-source) case, where all demands must be routed to a given root.
When $f$ is known, the problem becomes the well-studied single-sink buy-at-bulk (SSBaB) problem.  SSBaB is $NP$-hard---it generalizes the Steiner tree problem---but constant-factor approximations are known for any given $f$ (e.g.\ \cite{guha2001cfa, grandoni2010network}).  The special case where $f$ has the form $f(x) = \min\{x,M\}$ for some $M$ (edges can be ``rented'' for linear cost or ``bought'' for a fixed cost) is known as the single-sink rent-or-buy (SSRoB) problem and has also received significant attention (e.g.\ \cite{karger2000building, eisenbrand2010connected}).

Buy-at-bulk algorithms produce trees that are heavily tailored to the particular function at hand, but in some scenarios $f$ may be unknown or known to change over time.  
One setting where this arises is in aggregation of data in sensor networks.  The degree of redundancy among different sensor measurements may be unknown, or the same network may be used for aggregating different types of information with different amounts of redundancy.
In other situations rapid technological advancement may cause bandwidth costs to change drastically over time.
Further, in the interest of simplifying the design process and building robust networks, it may be useful to decouple the problem of designing the network topology from that of determining the exact characteristics of the information or goods flowing through that network.  In these settings it is desirable to find a single tree that is simultaneously good for all cost-functions, and from a theoretical perspective, the existence of such trees would reveal surprising structure in the problem.

There are two natural objective functions which capture the idea of simultaneous approximation for multiple cost functions.  Let $\mathcal{F}$ be the set of all concave, monotone non-decreasing cost functions satisfying $f(0) = 0$, $f(T)$ be the cost of a routing tree $T$ under function $f$, and $T_f^*$ be the optimal routing graph for $f$.
Note that due to the concavity of $f$ we may assume that $T_f^*$ is a tree.  
Let $\mathcal{R}$ be a randomized algorithm that returns a feasible routing tree $T$.
First, we could try to minimize the quantity
\begin{align}
\label{r1}
 \sup_{f \in \mathcal{F}} \frac{\E_{\mathcal{R}} [f(T)]}{f(T_f^*)}
\end{align}
which we call the {\em oblivious} approximation ratio.
If the oblivious ratio is small, then $\mathcal{R}$ returns a distribution that works well in expectation for any $f$.  
However, there may be no sample from this distribution that works for everything: for any tree $T$ there may be functions for which $f(T)$ is expensive.

To circumvent this problem we can work with the much stronger {\em simultaneous} approximation ratio.  For a deterministic or randomized algorithm $\mathcal{A}$ that returns a tree $T_{\mathcal{A}}$ the simultaneous ratio of $\mathcal{A}$ is defined as:
\begin{align}
\label{r2}
 \E_{\mathcal{A}} \left[\sup_{f \in \mathcal{F}} \frac{f(T_{\mathcal{A}})}{f(T_f^*)}\right]
\end{align}
A bound on the simultaneous ratio subsumes one on the oblivious ratio and proves there exists a single tree that is simultaneously good for all $f$. 

We emphasize that the distinction between the simultaneous and oblivious objectives is not a technicality in the objective but rather a fundamental difference and that the gap between these ratios can be large.  
Consider the problem of embedding arbitrary metrics into tree metrics, another problem that requires bounding the cost under many different functions (i.e.\ distortion of each edge).  
It is well-known that distributions over trees can achieve $O(\log n)$ expected distortion for all edges \cite{fakcharoenphol2003tba} but that even for simple graphs like the $n$-cycle no single tree can do better than $\Omega(n)$ distortion \cite{rabinovich1998lower}.  Therefore, the ratio between maximum expected distortion and the expected maximum distortion is $\Omega(n/\log n)$ in this case.

Goel and Estrin \cite{goel2003soc} introduced the problem of simultaneous SSBaB and gave an algorithm with an $O(\log D)$ bound on the simultaneous ratio \eqref{r2}, where $D$ is the total amount of demand.  
Goel and Post \cite{goel2009oblivious} recently improved the oblivious ratio \eqref{r1} to $O(1)$ for a large constant.  Trees for which the simultaneous ratio was $O(1)$ were not known to exist regardless of computation time.

In this paper we give the first constant guarantee on the simultaneous ratio, resolving the major open question of Goel and Estrin and Goel and Post \cite{goel2003soc, goel2009oblivious}.
Several aspects of our algorithm and analysis bear mentioning:
\begin{itemize}

\item We achieve a simultaneous approximation ratio of 47.45.  This is within a factor of 2.33 of the current best approximation for normal SSBaB of 20.42 \cite{grandoni2010network} and substantially smaller than the $O(1)$ oblivious bound \cite{goel2009oblivious}, which we estimate to be around 15 million.

\item The algorithm is entirely combinatorial, in contrast to the result of Goel and Post \cite{goel2009oblivious}, which uses separation oracles to prove an $O(1)$-oblivious distribution exists but reveals little of its structure.


\item Our analysis is short and simple, no more complex than the analysis of a normal SSBaB algorithm.

\item The runtime is only $O((t(n,m) + m + n\log n)\log D)$ for a graph with $n$ nodes, $m$ edges, and $D$ demand, where $t(n,m)$ is the runtime of an SSRoB approximation.

\end{itemize}

The algorithm is quite simple.  We first find approximate trees for a set of rent or buy basis functions,
prune this set to obtain a subset $L$ of trees whose total rent costs are increasing geometrically while total buy costs are dropping geometrically, and then prove it suffices to approximate every tree in $L$.  The set of bought nodes for each tree in $L$ defines a series of tree layers, which we stitch together using light approximate shortest-path trees (LASTs) \cite{khuller1995bms} to approximate both the minimum spanning tree (MST) and shortest-path tree.  
Finally, we consider any layer in the tree.  Using the geometrically changing costs and the properties of the LAST construction, we conclude that everything within the layer is an approximate MST, and everything outside approximates the shortest-path tree cost.

%

\subsection{Related Work}

The SSBaB problem was first posed by Salman et al.\ \cite{salman1997bbn}, and the first general approximation algorithm was given by Awerbuch and Azar \cite{awerbuch1997bbn}, who used metric tree embeddings \cite{Bartal96probabilistic} to achieve an $O(\log^2 n)$ ratio, later improved to $O(\log n)$ using better embeddings by Bartal \cite{bartal1998aam} and Fakcharoenphol et al. \cite{fakcharoenphol2003tba}.  Guha et al.\ \cite{guha2001cfa} gave the first constant approximation, and follow-up work by Talwar \cite{talwar2002ssb}, Gupta et al. \cite{gupta2007approximation}, Jothi and Raghavachari \cite{jothi2004iaa}, Grandoni and Italiano \cite{grandoni2006ias}, and Grandoni and Rothvo{\ss} \cite{grandoni2010network} has since reduced the constant to 20.42.  Most recent algorithms for SSBaB (and several related problems) are based on the sample and augment framework of Gupta et al.\ \cite{gupta2007approximation}.  Many algorithms using this framework have been derandomized by van Zuylen \cite{van2009deterministic}.

The special case of SSRoB has also been extensively studied, often as a special case of the connected facility location problem.  The first constant factor approximation was given by Ravi and Salman \cite{ravi1999approximation} as a special case of the traveling purchaser problem.  Karger and Minkoff \cite{karger2000building} gave an alternate algorithm and introduced connected facility location.  Gupta et al.\ improved the approximation to 9.01 \cite{gupta2001provisioning}, Swamy and Kumar to 4.55 \cite{swamy2004primal}, and Gupta et al.\ to 3.55 \cite{gupta2007approximation}.  Gupta et al.\ \cite{gupta2008cost} derandomized the 3.55-approximation to achieve a 4.2-approximation.  
Eisenbrand et al.\ \cite{eisenbrand2010connected} developed a randomized 2.92-approximation, which recently improved to 2.8 using the 1.39-approximation for Steiner tree of Byrka et al.\ \cite{byrka2010improved}.
Since we employ the 2.8-approximation, and the claimed ratio does not currently appear elsewhere in the literature, we present the brief calculation deriving this value in the appendix.
Both Williamson and van Zuylen \cite{williamson2007simpler} and Eisenbrand et al.\ \cite{eisenbrand2010connected} independently derandomized the 2.92-approximation to achieve a deterministic 3.28-approximation.

The problem of simultaneous approximation for multiple cost functions has been studied using both the oblivious and simultaneous objectives.  Goel and Estrin \cite{goel2003soc} were the first to explicitly pose the question of simultaneous approximations and gave an algorithm with an $O(\log D)$ simultaneous guarantee.  Prior to that Khuller et al.\ \cite{khuller1995bms} gave an algorithm to simultaneously approximate the two extreme cost functions $f(x) = 1$ and $f(x)=x$---a result which plays an important role in this paper---and metric tree embeddings had been applied to SSBaB \cite{awerbuch1997bbn, fakcharoenphol2003tba} to achieve an $O(\log n)$ bound for the oblivious objective. Enachescu et al.\ \cite{enachescu2005sfa} gave an $O(1)$ simultaneous guarantee for the special case of grid graphs with some spatial correlation.  Goel and Post \cite{goel2009oblivious} proved that an oblivious guarantee of $O(1)$ is achievable for all graphs.  Gupta et al.\ \cite{gupta2006ond} and Englert and R{\"a}cke \cite{englert2009oblivious} have studied several generalizations of the problem where both the demands and function are unknown, and multiple sinks are allowed.  In these settings the guarantee is generally $O(\log n)$ or $O(\textrm{polylog } n)$.

\section{Notation and Preliminaries}

Formally, we are given a graph $G = (V,E)$ with edge lengths $l_e$ for $e \in E$, a root node $r$, and a set of demand nodes $\mathcal{D} \subseteq V$ with integer demands $d_v$.  The total demand is $D = \sum_v d_v$.  We want to route $d_v$ flow from each $v$ to $r$ as cheaply as possible, where the cost of routing $x_e$ flow along edge $e$ is $l_ef(x_e)$ for some unknown, concave, monotone increasing function $f$ satisfying $f(0) = 0$.  Not knowing $f$, our objective is to find a feasible tree $T$ minimizing $\sup_f f(T)/f(T_f^*)$, where $T_f^*$ is the optimal graph for $f$.

We first show that we can restrict our analysis to a smaller class of basis functions.
Let $\epsilon >0$ be a small constant which will trade off the runtime and the approximation ratio, and $K = \lceil\log_{1+\epsilon} D\rceil$.  
For $0 \le i \le K$, define $M_i = (1+\epsilon)^i$, $A_i(x) = \min\{x,M_i\}$, and $T_i^*$ as the optimal tree for $A_i$.  
By the monotonicity and concavity of $f$, whenever $M_i \le x \le M_{i+1}$ we have $f(M_i) \le f(x) \le f(M_{i+1}) \le (1+\epsilon)f(M_i)$, so with a loss of only a factor of $1+\epsilon$ we can interpolate between $f(M_i)$ and $f(M_{i+1})$ and assume $f$ is piecewise linear with breakpoints only at powers of $1+\epsilon$.  
A nondecreasing concave function that is linear between powers of $1+\epsilon$ can be written as a nonnegative linear combination of $\{A_i\}_{0\le i \le K}$ by setting coefficients equal to the changes in slope: if the slope drops from $\delta_i$ to $\delta_{i+1}$ at $(1+\epsilon)^i$ it induces the term $(\delta_i-\delta_{i+1})A_i(x)$.  Now for a linear combination $\sum_i a_iA_i(x)$ and a tree $T$
\[
\frac{\sum_i a_iA_i(T)}{\sum_i a_iA_i(T_f^*)} \le \frac{\sum_i a_iA_i(T)}{\sum_i a_iA_i(T_i^*)} = \frac{\sum_i a_iA_i(T_i^*)\frac{A_i(T)}{A_i(T_i^*)}}{\sum_i a_iA_i(T_i^*)} \le \max_i \frac{A_i(T)}{A_i(T_i^*)}
\]
so it suffices to upper bound $\max_i A_i(T)/A_i(T_i^*)$.

We now define some notation and subroutines that will be important for our algorithm.
The problem of finding a good aggregation tree for the function $A_i(x) = \min\{x,M_i\}$ is an instance of the SSRoB problem, and we can find a $\lambda$-approximate tree $T_i$, where $\lambda$ is the best approximation ratio known, currently equal to $2.8$ using the algorithm of Eisenbrand et al.\ \cite{eisenbrand2010connected} and Byrka et al.\ \cite{byrka2010improved}.
We will assume the algorithm is deterministic.  If not (as in the case of the 2.8-approximation) we repeat it a polynomial number of times and pick the best tree, so we are close to a $\lambda$-approximation with very high probability.  In this case, our simultaneous approximation algorithm will have some tiny probability of failure.

The cost $A_i(T_i)$ can be broken into two pieces, the rent cost and the buy cost, based on whether $A_i$ is maxed out at $M_i$:

\begin{defn}
For an aggregation tree $T_i$ for cost function $A_i$ with $x_e$ flow on edge $e$, the {\em rent cost} $R_i$ and {\em normalized buy cost} $B_i$ are defined as
\begin{gather*}
R_i = \sum_{e \in T_i,\, x_e < M_i} l_eA_i(x_e) \\
B_i = \sum_{e \in T_i, \, x_e \ge M_i} l_e\frac{A_i(x_e)}{M_i} = \sum_{e \in T_i,\, x_e \ge M_i} l_e
\end{gather*}
\end{defn}

Note that edge costs composing $R_i$ are weighted by the amount of flow they carry, but edges in $B_i$ are not; they use unweighted edge costs.  The total cost of $T_i$ is given by $A_i(T_i) = R_i + M_iB_i$.  
The rent and buy costs also partition the nodes of $T_i$ into two sets:

\begin{defn}
The {\em core} $C_i$ of tree $T_i$ consists of $r$ and all nodes spanned by bought edges and the {\em periphery} contains all vertices outside $C_i$.
\end{defn}

If we condition on the nodes in $C_i$ then the rent-or-buy problem becomes easy: demands outside the core pay linear cost until they reach $C_i$, so they should take the shortest path, 
whereas within $C_i$ we pay a fixed cost per edge length, so the best strategy is to follow the min spanning tree.  
The cost $R_i$ is therefore at least the sum of shortest path distances to $C_i$, while $B_i$ is at least the weight of the MST of $C_i$.

In addition to the SSRoB approximation, we will also employ the light, approximate shortest-path tree algorithm of Khuller et al.\ \cite{khuller1995bms}:
\begin{defn}[\cite{khuller1995bms}]
For $\alpha \ge 1$ and $\beta \ge 1$, an {\em $(\alpha,\beta)$-light, approximate shortest-path} or {\em $(\alpha,\beta)$-LAST} is a spanning tree $T$ of $G$ with root $r$ such that
\begin{itemize}
\item For each vertex $v$, the distance from $v$ to $r$ in $T$ is at most $\alpha$ times the shortest path distance from $v$ to $r$ in $G$.
\item The edge weight of $T$ is at most $\beta$ times the weight of an MST of $G$.
\end{itemize}
\end{defn}
Khuller et al.\ show how to construct an $(\alpha,\beta)$-LAST for any $\alpha > 1$ and $\beta \ge \frac{\alpha+1}{\alpha-1}$.  
Roughly, the algorithm performs a depth-first traversal of the MST of $G$ starting from $r$, checking the stretch of the shortest path to each node.  If the path to some $v$ has blown up by at least an $\alpha$ factor, then it updates the tree to take the shortest path from $v$ to $r$, adjusting other tree edges and distances accordingly.  See the paper \cite{khuller1995bms} for a full description and analysis.

Finally, we define four parameters $\alpha$, $\beta$, $\gamma$, and $\delta$ used by our algorithm whose values we will optimize at the end.
\begin{itemize}
\item $\alpha > 1$ is the approximation ratio for shortest paths used in our LAST.

\item $\beta \ge \frac{\alpha+1}{\alpha-1}$ is the corresponding approximation to the MST in the LAST.

\item $\gamma > 1$ is the factor by which normalized buy costs $B_i$ increase from layer to layer in our tree.

\item $\delta > 1$ is the factor by which rent costs $R_i$ drop from layer to layer.

\end{itemize}
We now turn to a more thorough explanation of tree layers.

\section{Tree Layers}

In the normal SSBaB problem, the cost function is defined as $f(x) = \min_j \{\sigma_j + \delta_jx\}$, the cheapest of a collection of different ``pipes'' or ``cables'' given to the algorithm, each with an affine cost function $\sigma_j + \delta_jx$.  
It is common (e.g.\ \cite{guha2001cfa, gupta2007approximation}) to first prune these pipes to a smaller set with geometrically decreasing $\delta_j$'s and geometrically increasing $\sigma_j$'s and then build a solution in layers where each layer routes with a different pipe.

We perform an analogous procedure.  
We would like to build our simultaneous tree $T$ in a series of nested layers defined by the cores $C_i$ of each tree $T_i$, so that the core of $T$ under $A_i$ is similar to $C_i$, but we have no guarantees on the relationships between different cores: $C_i$ and $C_k$ may be entirely disjoint except for $r$.  
However, we will show that as long as normalized buy costs $B_i$ and $B_k$ are within a constant factor of each other, the same core can be used for both trees.
Consequently, we are able to define nested layers by choosing one $C_i$ for each order of magnitude of $B_i$.

After finding $\lambda$-approximate trees $T_i$ for each $A_i$, we loop through the costs $B_i$ and $R_i$, discarding $i$ whenever $B_i$ does not drop by $\gamma$ or $R_i$ does not grow by $\delta$.
We are left with a subset $L$ of the $T_i$ where the $B_i$'s are dropping by a factor of $\gamma$ and the $R_i$'s are growing by a factor of $\delta$.  The cores $C_i$ for each $i \in L$ will define the layers of our tree.
Algorithm \ref{layer_alg} describes the procedure in more detail.

\begin{algorithm}[htb]
\caption{Finding tree layers}
\label{layer_alg}
\LinesNumbered
\DontPrintSemicolon
\SetAlgoNoEnd
\SetAlgoNoLine
\SetKw{KwDownto}{down to}

\KwIn{Graph $G$ and demands $\mathcal{D}$}
\KwOut{Set $L$ and cores $C_i$ for each $i \in L$}

\For{$i \leftarrow 0$ \KwTo $K$} {
$T_i \leftarrow \lambda$-approximate tree for $A_i(x)$\;
}

\For{$i \leftarrow 1$ \KwTo $K$} { \label{first_loop}
\lIf {$A_{i}(T_{i-1}) < A_{i}(T_i)$} {
$T_{i} \leftarrow T_{i-1}$ \; \label{first_loop_end}
} }

\For{$i \leftarrow K-1$ \KwDownto $0$} { \label{second_loop}
\lIf {$A_{i}(T_{i+1}) < A_{i}(T_i)$} {
$T_i \leftarrow T_{i+1}$ \; \label{second_loop_end}
} }

\For{$i \leftarrow 0$ \KwTo $K$} {
calculate $C_i$, $B_i$, $R_i$\; \label{monotone_line}
}

\BlankLine

$L_B \leftarrow \emptyset$ \;
$B \leftarrow \infty$ \;

\For{$i \leftarrow 0$ \KwTo $K$} { \label{lb_start}
\If{$B_i < \frac{1}{\gamma}B$} {
$L_B \leftarrow L_B \cup \{i\}$ \;
$B \leftarrow B_i$ \; \label{lb_end}
} }

$L \leftarrow \emptyset$ \;
$R \leftarrow \infty$ \;
\ForEach {$i \in L_B$ in decreasing order} { \label{l_start}
\If{$R_i < \frac{1}{\delta}R$} {
$L \leftarrow L \cup \{i\}$ \;
$R \leftarrow R_i$ \; \label{l_end}
} }

\KwRet{$L$ and $C_i$ for each $i \in L$}

\end{algorithm}

Intuitively, as it becomes more expensive to buy edges the optimum will buy fewer edges and rent more.  In the case of approximations, the progression becomes muddled because for some $i$ the approximation guarantee may be tight while for $i+1$ we may get lucky and find the optimum, resulting in both rent and normalized buy costs dropping.  
We first show that the monotonicity in buy and rent costs still holds as long as each $T_i$ is better for $A_i$ than both $T_{i-1}$ and $T_{i+1}$.

\begin{lem}
\label{monotonicity_lemma}
After line \ref{monotone_line} of Algorithm \ref{layer_alg}, for every $i$ we have $B_i \ge B_{i+1}$ and $R_i \le R_{i+1}$.
\end{lem}

\begin{proof}
First we show that for each $i$, $A_i(T_i) \le \min\{A_i(T_{i+1}),A_i(T_{i-1})\}$.  After the loop on lines \ref{first_loop}--\ref{first_loop_end}, we have $A_i(T_i) \le A_i(T_{i-1})$, and after the second loop on lines \ref{second_loop}--\ref{second_loop_end} we have $A_i(T_i) \le A_i(T_{i+1})$, so we only need to show that the second loop does not break the first condition.  
If the second loop updates $T_i$ then $A_i(T_i)$ will only shrink, and if it changes $T_{i-1}$ it does this by setting $T_{i-1} \leftarrow T_i$ which preserves $A_i(T_i) \le A_i(T_{i-1})$.

Now consider $A_i(T_k)$ for any $k$.  By definition $A_i(x) \le x$ and $A_i(x) \le M_i$, so to upper bound $A_i(T_k)$ we may assume edges within $C_k$ pay $M_i$ per unit length, which sums to $M_iB_k$, and edges outside $C_k$ pay linear cost, or $R_k$ total, implying $A_i(T_k) \le R_k + M_iB_k$.
Therefore
\begin{gather*}
R_i + M_iB_i = A_i(T_i) \le A_i(T_{i+1}) \le R_{i+1} + M_iB_{i+1} \\
\Longrightarrow M_i(B_i - B_{i+1}) \le R_{i+1} - R_i
\end{gather*}
Similarly,
\begin{gather*}
R_{i+1} + M_{i+1}B_{i+1} = A_{i+1}(T_{i+1}) \le A_{i+1}(T_i) \le R_i + M_{i+1}B_i \\
\Longrightarrow R_{i+1} - R_i \le M_{i+1}(B_i - B_{i+1})
\end{gather*}
Combining the inequalities,
\[
M_i(B_i - B_{i+1}) \le R_{i+1} - R_i \le M_{i+1}(B_i - B_{i+1})
\]
If $B_i - B_{i+1} < 0$ the inequality is false because $M_{i+1} > M_i$, so we conclude $B_i \ge B_{i+1}$.  And using the first inequality, $0 \le M_i(B_i - B_{i+1}) \le R_{i+1} - R_i$, so $R_i \le R_{i+1}$.
\end{proof}

We need to show that we can restrict our attention to $T_i$ for $i \in L$.  
Suppose $i < k$ but $B_i \le \gamma B_k$.  
Using Lemma \ref{monotonicity_lemma}, observe that $A_k(T_i) \le R_i + M_kB_i \le R_k + \delta M_kB_k \le \delta A_k(T_k)$.  
Note this is independent of the size of the intersection of the cores $C_i$ and $C_k$ and any differences in routing.  
The following lemma generalizes this simple but key observation and proves that approximating each $i \in L$ is sufficient.

\begin{lem}
\label{structure_lemma}
Suppose there exists a tree $T$ and constants $c_B$ and $c_R$ such that for all $i \in L$ there exists a partition of the edges of $T$ into two sets $T_{B_i}$ and $T_{R_i}$ satisfying
\begin{itemize}
\item $A_0(T_{B_i}) \le c_B B_i$
\item $A_K(T_{R_i}) \le c_R R_i$
\end{itemize}
then for all $k \in \{0,\ldots,K\}$, $A_k(T) \le \max\{c_B\gamma, c_R\delta\}\lambda A_k(T_k^*)$.
\end{lem}

\begin{proof}
Let $k \in \{0,\ldots,K\}$.  
Let $j = \max \{j \in L_B | j\le k\}$.  
Either $j = k$ or $k$ was discarded due to $j$ on lines \ref{lb_start}--\ref{lb_end} because $B_k \ge \frac{1}{\gamma}B_j$, and either way $B_j \le \gamma B_k$.  
Now let $i = \min \{i \in L | i \ge j\}$.  
Again $i = j$ or $j$ was pruned due to $i$ on lines \ref{l_start}--\ref{l_end}, and $R_i \le \delta R_j$.  
Applying Lemma \ref{monotonicity_lemma} with $i \ge j$ and $j \le k$, we have $B_i \le B_j \le \gamma B_k$ and $R_i \le \delta R_j \le \delta R_k$.

This is sufficient to bound the cost of $A_k(T)$:
\begin{align*}
A_k(T) = A_k(T_{R_i}) + A_k(T_{B_i}) &\le A_K(T_{R_i}) + M_kA_0(T_{B_i}) \\
& \le c_R R_i + c_BM_k B_i \\
&\le c_R\delta R_k + c_B\gamma M_k B_k \\
& \le \max\{c_R\delta, c_B\gamma\}A_k(T_k) \le \max\{c_R\delta, c_B\gamma\}\lambda A_k(T_k^*)
\end{align*}
The equality follows because $T_{B_i}$ and $T_{R_i}$ partition the edges of $T$.  
The first inequality is because $A_K(x)$ and $M_kA_0(x)$ both upper bound $A_k(x)$, the second is by assumption, the third is from the derivation above, the fourth uses $A_k(T_k) = R_k + M_kB_k$, and the last is because $T_k$ is a $\lambda$-approximation.
\end{proof}

We will primarily assume that our SSRoB algorithm is a generic approximation, but Lemma \ref{structure_lemma} can easily be improved to take advantage of an SSRoB algorithm with a stronger guarantee that separately bounds $R_i$ and $M_iB_i$ in terms of the optimal costs $R_i^*$ and $M_iB_i^*$.
\begin{cor}
\label{bifactor_structure_lemma}
Let $T$, $c_B$, and $c_R$ be as in Lemma \ref{structure_lemma}, and suppose
\begin{itemize}
\item $R_i \le \mu_R R_i^* + \mu_B M_iB_i^*$
\item $M_iB_i \le \nu_R R_i^* + \nu_B M_iB_i^*$
\end{itemize}
then for all $k$, $A_k(T) \le \max\{c_R\delta\mu_R + c_B\gamma\nu_R, c_R\delta\mu_B + c_B\gamma\nu_B\}A_k(T_k^*)$.
\end{cor}

\begin{proof}
We change the inequalities in the proof above as follows:
\begin{align*}
A_k(T) \le c_R\delta R_k + c_B\gamma M_kB_k  & \le c_R\delta(\mu_R R_k^* + \mu_B M_kB_k^*) + c_B\gamma(\nu_R R_k^* + \nu_B M_kB_k^*) \\
& \le \max\{c_R\delta\mu_R + c_B\gamma\nu_R, c_R\delta\mu_B + c_B\gamma\nu_B\}A_k(T_k^*)
\end{align*}
\end{proof}

\section{Constructing the Tree}

The construction of the tree itself is quite simple.  We have a set of indices $L$ and core sets $C_i$ for $i \in L$.  Starting with the largest $i \in L$, i.e.\ smallest $B_i$, and working downward, we connect each $C_i$ to $T$, the tree so far, with a LAST.  
Algorithm \ref{tree_alg} describes the procedure more formally.
The notation $G/T$ represents contracting $T$ to a single node in $G$, and $G[C_i]$ is the induced subgraph on $C_i$, so $(G/T)[C_i]$ denotes first contracting $T$ and then restricting to nodes in $C_i$.

\begin{algorithm}[htb]
\caption{Constructing the tree}
\label{tree_alg}
\LinesNumbered
\DontPrintSemicolon
\SetAlgoNoEnd
\SetAlgoNoLine

\KwIn{Graph $G$.  Set $L$ and accompanying $C_i$ for each $i \in L$}
\KwOut{Aggregation tree $T$}

$T \leftarrow \{r\}$ \;

\ForEach{$i \in L$ in decreasing order} {
$T' \leftarrow$ $(\alpha, \beta)$-LAST of $(G/T)[C_i]$ with root $T$ \;
$T \leftarrow T \cup T'$ \;
}

\KwRet{$T$}

\end{algorithm}


\begin{lem}
The graph $T$ constructed by Algorithm \ref{tree_alg} is a tree and spans all demand nodes.
\end{lem}

\begin{proof}
Observe that $0 \in L$, and $R_0 = 0$, so $C_0$ covers all demands.  Therefore after the last iteration $T$ spans $\mathcal{D}$.  
Each iteration only adds edges spanning new vertices, so no cycles are created.
\end{proof}

The tree may contain paths connecting Steiner nodes that carry no flow.  Such edges can be safely pruned or just ignored because they contribute nothing to the cost.

All that remains is to define the partitions $T_{B_i}$ and $T_{R_i}$ and prove the bounds needed in Lemma \ref{structure_lemma}.
The set $T_{B_i}$ contains all edges present in $T$ after connecting $C_i$, and $T_{R_i}$ contains the rest.  Both cost bounds will follow easily from the geometrically changing costs:
the cost $A_0(T_{B_i})$ is dominated by the cost of the $C_i$ layer, an approximate MST, and $A_K(T_{R_i})$ is dominated by the rent costs of the next layer, an approximate shortest-path tree.
First, we bound the normalized buy cost of $T_{B_i}$:

\begin{lem}
\label{buy_lemma}
Let $i \in L$, and $T_{B_i}$ be the tree $T$ after the round when $C_i$ is added.  Then the edge cost $A_0(T_{B_i})$ is at most $\frac{\beta\gamma}{\gamma-1}B_i$.
\end{lem}

\begin{proof}
The proof is by decreasing induction on $i$, i.e.\ in the order in which the layers are built.
Let $c$ be a constant to be chosen at the end.
The base case is the largest $i \in L$, which is the smallest $i$ such that $B_i = 0$.  In this case, $C_i = \{r\}$, $T_{B_i} = \{r\}$, and the edge cost is 0.

Now let $i \in L$, $k = \min \{k \in L | k > i\}$ be the previous (inner) layer, and suppose the edge cost of $T_{B_k}$ is at most $cB_k$.  
By the construction of $L$, we know $B_k < \frac{1}{\gamma}B_i$, implying $T_{B_k}$ costs at most $\frac{c}{\gamma}B_i$.
The cost of an MST of $C_i$ in $G$ is at most $B_i$, and $T_{B_k}$ may already span part of $C_i$, so connecting the rest with an MST\footnote{We could allow Steiner nodes and use a Steiner tree approximation, but this would not improve the worst-case bound.} of $(G/T_{B_k})[C_i]$ costs at most $B_i$.  
Using an $(\alpha, \beta)$-LAST scales the cost by at most $\beta$.

The total cost of edges laid so far is at most $\frac{c}{\gamma}B_i + \beta B_i$, so the proof is complete as long as $cB_i \ge \frac{c}{\gamma}B_i + \beta B_i$.  Set $c = \frac{\beta\gamma}{\gamma - 1}$:
 \[
 c \ge \beta + \frac{c}{\gamma} \Longleftrightarrow c\left(1-\frac{1}{\gamma}\right) \ge \beta \Longleftrightarrow c \ge \frac{\beta\gamma}{\gamma - 1}
 \]
\end{proof}

Now we bound the rent costs of $T_{R_i}$ by a similar proof.

\begin{lem}
\label{rent_lemma}
Let $i \in L$ and $T_{R_i} = T/T_{B_i}$, i.e.\ all edges outside of $T_{B_i}$.  Then the rent cost $A_K(T_{R_i})$ is at most $\frac{\alpha\delta}{\delta - \alpha - 1}R_i$.
\end{lem}

\begin{proof}
We prove by increasing induction on $i \in L$ (the reverse of Lemma \ref{buy_lemma}) that $A_K(T_{R_i}) \le cR_i$ for some $c$ to be determined.
Since $0 \in L$, and $T_{B_0}$ covers everything, the base case $T_{R_0}$ costs 0 too.

For the inductive case, let $i \in L$, $k = \max\{k \in L | k < i\}$ be the next (outer) layer, and $T_{R_k}$ have rent cost at most $cR_k$.
As before, note $R_k < \frac{1}{\delta}R_i$, so $A_K(T_{R_k}) \le \frac{c}{\delta}R_i$.
Tree $T_{B_i}$ spans $C_i$ and possibly more, so if all demands outside $T_{B_i}$ took the shortest path from their sources to $T_{B_i}$ the shortest-path cost would be at most $R_i$.
However, the edges of $T_{R_k}$ have moved some demands around, and by the time they reach the current layer they may be farther from $T_{B_i}$ than their original sources were.  But by the triangle inequality the cost of sending all demands from their current locations to $T_{B_i}$ via shortest paths is at most $\frac{c}{\delta}R_i + R_i$, the cost of sending all flow in $T_{R_k}$ back to its source and from there to $T_{B_i}$ using shortest paths.  The LAST algorithm guarantees $\alpha$-approximate shortest paths, multiplying the cost by $\alpha$.

Consequently the total rent cost for $T_{R_i}$ is bounded by $\frac{c}{\delta}R_i + \alpha\left(\frac{c}{\delta}R_i + R_i\right)$, which needs to be at most $cR_i$.  We can set $c = \frac{\alpha\delta}{\delta-\alpha-1}$:
\[
c \ge \frac{c}{\delta} + \frac{c\alpha}{\delta} + \alpha
\Longleftrightarrow c\left(1-\frac{1}{\delta}-\frac{\alpha}{\delta}\right) = c\frac{\delta - \alpha - 1}{\delta} \ge \alpha
\Longleftrightarrow c \ge \frac{\alpha\delta}{\delta - \alpha -1}
\]
\end{proof}

We note that Lemma \ref{rent_lemma} explains how we circumvent a major obstacle to an $O(1)$-simultaneous approximation---the $\Omega(\log n)$ distortion lower bound for embedding arbitrary metrics into tree metrics \cite{Bartal96probabilistic}.  
If we needed to maintain distances between many pairs of nodes the task would be hopeless, but Lemma \ref{rent_lemma} shows that it suffices to preserve the distance of each node to the next layer, so the graph of distances to be maintained forms a tree.

We can now complete the proof of our main theorem and choose the optimal parameters.

\begin{thm}
\label{main_thm}
The tree $T$ achieves a simultaneous approximation ratio of $(1+\epsilon)\lambda(8+4\sqrt{5})$ using a $\lambda$-approximation to SSRoB.  In particular,
\begin{itemize}
\item There is a randomized polynomial time algorithm that finds a 47.45 simultaneous approximation with high probability.
\item There is a deterministic polynomial time algorithm that finds a 55.58 simultaneous approximation.
\item There exists a tree that is a 16.95 simultaneous approximation.
\end{itemize}
\end{thm}

\begin{proof}
Applying Lemma \ref{structure_lemma} with $c_B = \frac{\beta\gamma}{\gamma-1}$ (Lemma \ref{buy_lemma}) and $c_R = \frac{\alpha\delta}{\delta - \alpha - 1}$ (Lemma \ref{rent_lemma}), the final approximation ratio for an arbitrary cost function $f$ is
\begin{align*}
(1+\epsilon)\lambda \max\left\{ \frac{\beta\gamma^2}{\gamma-1}, \frac{\alpha\delta^2}{\delta-\alpha-1}\right\}
\end{align*}
where the extra $1+\epsilon$ comes from the approximation of $f$ by a combination of $A_i$'s.  Now it is a simple matter of applying calculus to find the optimal values for $\alpha$, $\beta$, $\gamma$, and $\delta$.  We set
\begin{align*}
\alpha &= \frac{1+\sqrt{5}}{2} & \beta &= 2+\sqrt{5} & \gamma &= 2 & \delta &=3+\sqrt{5}
\end{align*}
for which $\frac{\beta\gamma^2}{\gamma-1} = \frac{\alpha\delta^2}{\delta-\alpha-1}$.
The derivation of these values is presented in the appendix.

This gives us $\frac{\beta\gamma^2}{\gamma-1} = 4(2+\sqrt{5})$ so the simultaneous approximation ratio is
\[
(1+\epsilon) \lambda (8+4\sqrt{5}).
\]
Now,
\begin{itemize}
\item Using the best randomized approximation $\lambda = 2.8$, and the ratio is 47.45 with high probability.
\item Using the best deterministic approximation $\lambda = 3.28$, and the ratio is 55.58.
\item If the algorithm is allowed to run in exponential time $\lambda = 1$, and the ratio is 16.95.
\end{itemize}
\end{proof}

The 2.8-approximation of Eisenbrand et al.\ \cite{eisenbrand2010connected} actually provides a slightly stronger guarantee on $R_i$ and $B_i$, and we can use Corollary \ref{bifactor_structure_lemma} to get a tiny improvement in the approximation ratio at the cost of a more complex derivation:

\begin{thm}There is a randomized polynomial time algorithm that finds a 47.07 simultaneous approximation with high probability.
\label{bifactor_main_thm}
\end{thm}

\begin{proof}
Lemma 2 and Theorem 5 in \cite{eisenbrand2010connected} prove that
\begin{align*}
\E[R_i] \le 2R_i^* + \frac{.807}{x}M_iB_i^* && \E[M_iB_i] \le \rho(x+\epsilon)R_i^* + \rho M_iB_i^*
\end{align*}
where $\rho=1.39$ is the Steiner tree approximation ratio and $x\in (0,1]$ is a parameter.
Applying Corollary \ref{bifactor_structure_lemma} with
\begin{align*}
\mu_R = 2 && \mu_B = \frac{.807}{x} && \nu_R = \rho(x+\epsilon) && \nu_B = \rho
\end{align*}
the simultaneous ratio is bounded by
\[
(1+\epsilon)\max\left\{\frac{2\alpha\delta^2}{\delta-\alpha-1}+\frac{\rho(x+\epsilon)\beta\gamma^2}{\gamma-1},\frac{.807\alpha\delta^2}{x(\delta-\alpha-1)} + \frac{\rho\beta\gamma^2}{\gamma-1}\right\}
\]
Using
\begin{align*}
x= .5995 && \alpha = 1.5495 && \beta = \frac{\alpha+1}{\alpha-1} && \gamma=2 && \delta = 2\alpha+2
\end{align*}
yields a simultaneous ratio of 47.07.  The parameters are derived in the appendix. 
\end{proof}

We leave as an open question the problem of exploiting Corollary \ref{bifactor_structure_lemma} to substantially improve the ratio.

%
%

\subsection{Runtime}

Let $t(n,m)$ be the running time of our SSRoB approximation on a graph with $n$ vertices and $m$ edges, which must be at least $\Omega(n)$ to write down the output.  When $\epsilon$ is constant, running the SSRoB approximation for each $i$ takes $O(t(n,m)\log D)$.  Subsequent loops in Algorithm \ref{layer_alg} take $O(n\log D) = O(t(n,m)\log D)$.  

For each of the $O(\log D)$ iterations of Algorithm \ref{tree_alg} we need to do a graph contraction and run the LAST algorithm, which requires computing the MST and shortest path trees.  The computation of the shortest path tree takes $O(m + n\log n)$ and dominates the other steps.  Combining the two algorithms, the total time is $O((t(n,m) + m + n\log n)\log D)$.

\section{Open Problems}

We have answered the open questions posed by Goel and Estrin and Goel and Post \cite{goel2003soc, goel2009oblivious}, but there are several avenues for further work.  Our simultaneous ratio of 47.45 already surpasses many algorithms for normal SSBaB and is only a factor of 2.33 away from the best.  
It would be nice to eliminate this gap or, alternately, prove that a gap exists between the approximation achievable for fixed $f$ and the best simultaneous ratio.  We know of no lower bounds on what simultaneous ratio may be possible, so any progress in this direction would also be interesting.
Generalizing the settings in which $O(1)$ simultaneous ratios are possible would be interesting, but may be unlikely given that one must contend with lower bounds for metric tree embedding \cite{Bartal96probabilistic} and multi-sink buy-at-bulk \cite{andrews2004hardness}.

\section*{Acknowledgements}
We thank the anonymous reviewers for many helpful comments that improved the presentation and for suggesting Corollary \ref{bifactor_structure_lemma} and Theorem \ref{bifactor_main_thm}.

\pdfbookmark[1]{\refname}{My\refname}
\bibliographystyle{alphaurl_inline}
\bibliography{r_2}

\appendix

\section{Derivation of optimal parameters}

Here we derive the optimal values for the parameters used in the proofs of Theorems \ref{main_thm} and \ref{bifactor_main_thm}.

\begin{proof}[Parameters for Theorem \ref{main_thm}]
We need to minimize the expression
\begin{align}
\label{approx_ratio}
\max\left\{ \frac{\beta\gamma^2}{\gamma-1}, \frac{\alpha\delta^2}{\delta-\alpha-1}\right\}
\end{align}
The easiest parameters to fix are $\gamma$ and $\delta$.  For $\gamma$:
\[
\frac{d}{d\gamma}\left[\frac{\beta\gamma^2}{\gamma-1}\right] = \beta\frac{(2\gamma)(\gamma-1) - \gamma^2}{(\gamma-1)^2} = 0
\Longrightarrow \gamma(\gamma-2) = 0 \Longrightarrow \gamma = 2
\]
For $\delta$:
\begin{gather*}
\frac{d}{d\delta}\left[\frac{\alpha\delta^2}{\delta-\alpha-1}\right] = 
\alpha\frac{2\delta(\delta-\alpha-1) - \delta^2}{(\delta-\alpha-1)^2} = 0 \\
\Longrightarrow \delta^2 - 2\alpha\delta - 2\delta = \delta(\delta - 2\alpha -2) = 0
\Longrightarrow \delta = 2\alpha + 2
\end{gather*}

Plugging $\gamma = 2$ and $\delta = 2\alpha+2$ into \eqref{approx_ratio}, $\beta\frac{\gamma^2}{\gamma-1} = 4\beta$, and
\[
\frac{\alpha\delta^2}{\delta-\alpha-1} = \frac{\alpha(2\alpha+2)^2}{(2\alpha+2)-\alpha-1} 
= 4\alpha(\alpha+1)
\]
so \eqref{approx_ratio} is now $\max\{4\beta,4\alpha(\alpha+1)\}$.  The constraints on $\alpha$ and $\beta$ require $\beta \ge \frac{\alpha+1}{\alpha-1}$ \cite{khuller1995bms}, so one term blows up if the other shrinks.  To minimize the maximum set the two expressions to be equal:
\[
\beta = \frac{\alpha+1}{\alpha-1} = \alpha(\alpha+1)
\Longrightarrow \alpha(\alpha-1) = 1 
\Longrightarrow \alpha^2-\alpha-1 = 0 \Longrightarrow \alpha = \frac{1 \pm \sqrt{5}}{2}
\]
Using $\alpha = \frac{1+\sqrt{5}}{2}$, we get
\[
\beta = \frac{\alpha+1}{\alpha-1} = \frac{3+\sqrt{5}}{-1+\sqrt{5}} = \frac{(3+\sqrt{5})(1+\sqrt{5})}{4}
= \frac{8 + 4\sqrt{5}}{4} = 2+\sqrt{5}
\]
and $\delta = 2\alpha+2 = 3+\sqrt{5}$.
\end{proof}

\begin{proof}[Parameters for Theorem \ref{bifactor_main_thm}]
We need to minimize
\begin{align}
\label{bifactor_approx_expression}
\max\left\{\frac{2\alpha\delta^2}{\delta-\alpha-1}+\frac{\rho(x+\epsilon)\beta\gamma^2}{\gamma-1},\frac{.807\alpha\delta^2}{x(\delta-\alpha-1)} + \frac{\rho\beta\gamma^2}{\gamma-1}\right\}
\end{align}

For $\gamma$ and $\delta$, both expressions inside the $\max$ function are minimized exactly as above in Theorem \ref{main_thm} with $\gamma = 2$ and $\delta = 2\alpha+2$.  Also as in Theorem \ref{main_thm}, $\beta$ can be set to $\frac{\alpha+1}{\alpha-1}$, the minimum allowed by the constraints.  Plugging in these values and simplifying, expression \ref{bifactor_approx_expression} becomes
\[
4(\alpha+1)\max\left\{2\alpha+\frac{\rho(x+\epsilon)}{\alpha-1}, \frac{.807\alpha}{x} + \frac{\rho}{\alpha-1}\right\}
\]
Set the two expressions inside the maximization to be equal, and solve the resulting quadratic equation in $x$ to get (for $\epsilon=0$):
\[
x = \frac{1}{2} - \frac{\alpha(\alpha-1)}{\rho} + \sqrt{\frac{\alpha^2(\alpha-1)^2}{\rho^2} - \frac{.193\alpha(\alpha-1)}{\rho} + \frac{1}{4}}
\]
using that $x > 0$.

The problem is now to minimize
\[
4(\alpha+1)\left(2\alpha + \frac{\rho x}{\alpha-1}\right) = 4\alpha(\alpha+1) + \frac{2\rho(\alpha+1)}{\alpha-1} + 4(\alpha+1)\sqrt{\alpha^2 - \frac{.193\rho\alpha}{\alpha-1} + \frac{\rho^2}{(\alpha-1)^2}}
\]
This expression is unwieldy to optimize analytically but when $\rho=1.39$ it achieves a minimum of about 47.068 for $\alpha \approx 1.5495$.  When $\alpha=1.5495$, $x \approx .5995$.
\end{proof}

\section{Approximation ratio for SSRoB}
\label{appendix_ratios}

The current-best algorithm for SSRoB \cite{eisenbrand2010connected} depends on the approximation ratio for Steiner tree, which has recently been reduced to 1.39 \cite{byrka2010improved}.  The improved SSRoB ratio does not currently appear in the literature, so we include the calculation for completeness.  See the original paper for details.

\begin{thm}[\cite{eisenbrand2010connected, byrka2010improved}]
There is a 2.8-approximation for SSRoB.
\end{thm}

\begin{proof}
By the proof of Theorem 6 in \cite{eisenbrand2010connected}, there is an SSRoB algorithm with expected cost
\[
\rho(MB^* + (x+\epsilon)R^*) + 2R^* + 0.807\frac{MB^*}{x}
\]
where $\rho$ is the approximation ratio for Steiner tree and $x$ is a parameter in $(0,1]$.  Set the coefficients of $R^*$ and $MB^*$ to be equal and solve the resulting quadratic equation in $x$.  For $\rho=1.39$, choosing $x = .5735$ gives an approximation ratio of 2.80.
\end{proof}

\end{document}